\documentclass{llncs}
\usepackage{graphicx}
\usepackage{mathptmx}      

\usepackage[T1]{fontenc}
\usepackage[misc]{ifsym}
\usepackage{multirow}

\usepackage{tikz}
\usepackage{version}
\usetikzlibrary{patterns}
\usepackage{subcaption}
\usepackage{pifont}
\usepackage{stmaryrd}
\usepackage{amsmath}
\usepackage{amssymb}
\usepackage{todonotes}
\usepackage{bussproofs}
\usepackage[shortlabels]{enumitem}
\usepackage{float}
\usepackage{multicol}
\usepackage{booktabs} 
\usepackage{array}
\usepackage[noend]{algpseudocode}
\usepackage{algorithm}
\usepackage{cancel}
\usepackage{tcolorbox}
\usepackage{listings}
\usepackage{setspace}
\usepackage{booktabs}
\usepackage{ifthen}
\usepackage{needspace}

\usepackage{eucal}

\usepackage[
   a4paper,
   pdftex,
   pdftitle={Superposition with Delayed Unification},
   pdfauthor={Ahmed Bhayat, Johanes Schoisswohl, and Michael Rawson}
]{hyperref}
\hypersetup{final}

\DeclareFontFamily{OT1}{pzc}{}
\DeclareFontShape{OT1}{pzc}{m}{it}{<-> s * [1.10] pzcmi7t}{}
\DeclareMathAlphabet{\mathcalx}{OT1}{pzc}{m}{it}

\renewcommand{\inf}[1]{\textsc{#1}}

\DeclareMathAlphabet{\mathpzc}{OT1}{pzc}{m}{it}

\newcommand{\tuple}[1]{ \overline{#1} }

\newcommand{\fosubterm}[2]{#1 [ #2 ]}

\newcommand{\eq}{\approx}
\newcommand{\noteq}{\not\eq}
\newcommand{\eqOrNorEq}{\, \dot{\eq} \,}
\newcommand{\constraints}{\mathpzc{CS}}
\newcommand{\infsys}{\mathit{Inf}}
\newcommand{\ginfsys}{\mathit{GInf}}
\newcommand{\ginfgsel}{\mathit{GInf}^{\gsel}}
\newcommand{\ggsel}{\G^{gsel}}
\newcommand{\concl}{\mathit{concl}}
\newcommand{\prems}{\mathit{prems}}
\newcommand{\mprem}{\mathit{mprem}}

\newcommand{\universes}{\mathcal{U}}
\newcommand{\interpfun}{\mathcal{I}}
\newcommand{\model}{\mathpzc{M}}
\newcommand{\denotes}[3]{\llbracket #1 \rrbracket^{#2}_{#3} }

\newcommand{\holdsin}[2]{#1 \models #2}

\newcommand{\issubsumedeq}{\mathrel{\raisebox{-0.2pt}{\Large\rlap{\kern0.5pt$\cdot$}}}\geq}
\newcommand{\subsumeseq}{\mathrel{\raisebox{-0.2pt}{\Large\rlap{\kern3.2pt$\cdot$}}}\leq}
\newcommand{\issubsumed}{\mathrel{\raisebox{-1.1pt}{\Large\rlap{\kern0.5pt$\cdot$}}}>}
\let\oldsqsubset\sqsubset
\renewcommand{\sqsubset}{%
  \mathrel{\raisebox{1pt}{$\oldsqsubset$}}%
}

\newcommand{\G}{\mathpzc{G}}
\newcommand{\GSigmaN}{\mathpzc{G}(N)} 
\newcommand{\GSigma}[1]{\mathpzc{G}(#1)} 
\newcommand{\redclauses}{\mathit{Red}_{Cl}(N)}
\newcommand{\redc}{\mathit{Red}_{Cl}}
\newcommand{\gredclauses}{\mathit{GRed}_{Cl}(N)}
\newcommand{\gredc}{\mathit{GRed}_{Cl}}
\newcommand{\redinfs}{\mathit{Red}_I(N)}
\newcommand{\redi}{\mathit{Red}_I}
\newcommand{\gredinfs}{\mathit{GRed}^{\gsel}_I(N)}
\newcommand{\gredi}{\mathit{GRed}^{\gsel}_I}
\newcommand{\sel}{\mathit{sel}}
\newcommand{\gsel}{\mathit{gsel}}
\newcommand{\Ginv}{\mathpzc{G}^{-1}}

\newcommand{\orr}{\, \vee \,}

\newtcolorbox{conventionbox}[1]{colback=blue!3!white,colframe=blue!100!black,fonttitle=\bfseries,title=#1}

\newenvironment{scprooftree}[1]%
  {\gdef\scalefactor{#1}\begin{center}\proofSkipAmount \leavevmode}%
  {\scalebox{\scalefactor}{\DisplayProof}\proofSkipAmount \end{center} }

\def\orcidID#1{\href{http://orcid.org/#1}{\raisebox{-1.25pt}{\includegraphics{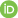}}}}

\title{Superposition with Delayed Unification}
\titlerunning{Superposition with Delayed Unification}
\authorrunning{Bhayat, Schoisswohl and Rawson}
\author{
Ahmed Bhayat \inst{1}\orcidID{0000-0002-1343-5084}
\and
Johannes Schoisswohl \inst{2}\orcidID{0000-0001-5550-196X}
\and
Michael Rawson \inst{2}\orcidID{0000-0001-7834-1567}
}
\institute{
University of Manchester, Manchester, UK
\email{ahmed.bhayat@manchester.ac.uk}
\and
TU Wien, Vienna, AT \\
\email{\{michael.rawson,johannes.schoisswohl\}@tuwien.ac.at}
}

\begin{document}
\maketitle
\begin{abstract}
Classically, in saturation-based proof systems, unification has been considered atomic.
However, it is also possible to move unification to the calculus level, turning the steps of the unification algorithm into inferences. 
For calculi that rely on unification procedures returning large or even infinite
sets of unifiers, integrating unification into the calculus is an attractive method of dovetailing unification and
inference. This applies, for example, to AC-superposition and higher-order superposition. We show
that first-order superposition remains complete when moving unification rules to the calculus level. We discuss some
of the benefits this has even for standard first-order superposition and provide an experimental evaluation.

\end{abstract}
\section{Introduction}

Unification is a key feature in many proof calculi, particularly those based on the saturation framework.
It acts as a filter, reducing the number of inferences that need to be carried out by instantiating terms only to the degree necessary.
However, many unification algorithms have large time complexities and produce large, or even infinite, sets of unifiers.
This is the case, for example, for AC-unification, which can produce a doubly exponential number of unifiers \cite{domenjoud-1992-technical}, and higher-order unification, which can produce an infinite set of unifiers \cite{snyder-1989-higher}.
This motivates the study of how unification rules can be integrated into proof calculi to allow them to dovetail with standard calculus rules.
One way to achieve this is to use the concept of unification with abstraction \cite{reger-unif-with-abstarction-2018,schoisswohl-2023-ALASCA}.
The general idea is that during the unification process, instead of solving all unification pairs, certain pairs are retained and added to the conclusion of an inference as negative \emph{constraint} literals.
Calculus-level unification inferences then work on such literals to solve these constraints and remove the literals in the case they are unifiable. 
Note how this differs from constrained resolution-style calculi such as \cite{bachmair-1992-basic,nieuwenhuis-1992-basic-complete} where the constraints are completely separate from the rest of the clause and are not subject to inferences.

To demonstrate the idea of dedicated unification inferences in combination with unification with abstraction, we provide the following example.

\[ C_1 = f(g(a,x)) \noteq t \qquad C_2 = f(g(a,b)) \eq t \]

\noindent
A standard superposition calculus would proceed by unifying $f(g(a,b))$ and $f(g(a,x)$ with the unifier $\sigma = \{x \rightarrow b \}$ and then rewriting $C_1$ with $C_2$ to derive $t\sigma \noteq t\sigma$. 
Equality resolution on $t\sigma \noteq t\sigma$ would then derive $\bot$. 
It is also possible to proceed by rewriting $C_1$ with $C_2$ \emph{without} computing $\sigma$ and instead add the constraint literal $g(a,x) \noteq g(a,b)$ to the conclusion to derive $t \noteq t \orr g(a,x) \noteq g(a,b)$. 
A dedicated unification inference could then decompose the constraint literal resulting in $t \noteq t \orr a \noteq a \orr b \noteq x$. 
Further unification inferences could bind $x$ to $b$, and remove the trivial pairs $a \noteq a$ and $t \noteq t$ to derive $\bot$. 

In this paper, we investigate moving unification to the calculus level for standard first-order superposition.
Whilst this may seem like a regressive step, as we lose much of unification's power to act as a filter on inferences and hence produce many more clauses, we think the investigation is valuable for two reasons. 

Firstly, by showing how syntactic first-order unification can be lifted to the calculus level, we provide a roadmap for how more complex unification problems can be lifted to the calculus level.
This may prove particularly useful in the higher-order case, where abstraction may expose terms to standard calculus rules that were unavailable before.
Moreover, we note that in our calculus we do not turn the entire unification problem into a constraint, but rather a subproblem. Whilst this may be merely an interesting detail for first-order unification, for more complex unification problems, such a method could be used to eagerly solve simple unification subproblems whilst delaying complex subproblems by adding them as constraints.

Secondly, one of the most expensive operations in first-order theorem provers is the maintenance of indices. 
Indices are crucial to the performance of modern solvers, as they facilitate the efficient retrieval of terms unifiable or matchable with a query term.
However, solvers typically spend a large amount of time inserting and removing terms from indices as well as unifying against terms in the indices.
This is particularly the case in the presence of the AVATAR architecture \cite{voronkov-2014-avatar} wherein a change in the model can trigger the insertion and removal of thousands of terms from various indices.
By moving unification to the calculus level, we can replace complex indices with simple hash maps, since to trigger an inference we merely need to check for top symbol equality and not unifiability. Insertion and deletion become $O(1)$ time operations.
However, for first-order logic, we do not expect the time gained to offset the downsides of extra inferences carried out and extra clauses created.
Our experimental results back up this hypothesis (see Section \ref{sec:results}).
Our main contributions are:

\begin{itemize}
\item Designing a modified superposition calculus that moves unification to the calculus level (Section \ref{sec:calculus}).
\item Proving the calculus to be statically and dynamically refutationally complete (Section \ref{sec:completeness_proof}).
\item Providing a thorough empirical evaluation of the calculus (Section \ref{sec:results}).
\end{itemize}

\section{Preliminaries}
\label{sec:tech_background}

\paragraph{Syntax}
We consider standard monomorphic first-order logic with equality. 
We assume a signature consisting of a finite set of (monomorphically) typed function symbols and a single predicate, equality, denoted by $\eq$.
A non-equality atom $A$ can be expressed using equality as $A \eq \top$ where $\top$ is a special function symbol \cite{schulz-2002-braniac}.
Terms are formed in the normal way from variables and function symbols.
We commonly use $s$, $t$ or $u$ or their primed variants to refer to terms.
We write $s : \tau$ to show that term $s$ has type $\tau$.
A term is ground if it contains no variables.
We use the notation $\tuple{s}_n$ to refer to a tuple or list of terms of length $n$.
More generally, we use the over bar notation to refer to tuples and lists of various objects.
Where the length of the tuple or list is not relevant, we drop the subscript.
By $s_i$ we denote the $i$th element of the tuple $\tuple{s}_n$.
Literals are positive or negative equalities written as $s \eq t$ and $s \noteq t$ respectively.
We use $s \eqOrNorEq t$ to refer to either a positive or a negative equality.
Clauses are multisets of literals. 
A clause that contains no literals is known as the empty clause and denoted by $\bot$.

A substitution is a mapping from variables to terms.
We assume, w.l.o.g.,  that all substitutions are idempotent.
We commonly denote substitutions using $\sigma$ and $\theta$ and denote the application of a substitution $\sigma$ to a term $s$
by $s\sigma$. 
A substitution $\theta$ is grounding for a term $s$, if $s\theta$ is ground.
The definition of grounding substitution can be extended to literals and clauses in the obvious manner.
A substitution $\sigma$ is a unifier of terms $s$ and $t$ if $s\sigma = t\sigma$.
A unifier $\sigma$ is more general than a unifier $\sigma'$ if there exists a substitution $\rho$ such that $\sigma\rho = \sigma'$.
With respect to syntactic first-order unification, if two terms are unifiable then they have a single most general unifier 
up to variable naming \cite{handbook-unification-theory}.

A transitive irreflexive relation over terms is known as an ordering.
The superposition calculus we present below is, as usual, parameterised by a simplification ordering on ground terms.
An ordering $\succ$ is a simplification ordering, if it possesses the following properties.
It is total on ground terms.
It is compatible with contexts, meaning that if $s \succ t$, then $\fosubterm{u}{s} \succ \fosubterm{u}{t}$.
It is well-founded.
Note that every simplification ordering has the subterm property. 
Namely, that if $t$ is a proper subterm of $s$, then $s \succ t$. 
For non-ground terms, the only property that is required of the ordering is that it is stable under substitution.
That is, if $s \succ t$ then for all substitutions $\sigma$, $s\sigma \succ t\sigma$.
We extend the ordering $\succ$ to literals in the standard fashion via its multiset extension. 
A positive literal $s \eq s'$ is treated as the multiset $\{s,s'\}$, whilst a negative literal $s \noteq s'$ is treated as the multiset $\{s,s,s',s' \}$.         
The ordering is extended to clauses by its two-fold multiset extension.
We use $\succ$ to denote the ordering on terms and its multiset extensions to literals and clauses.
\paragraph{Semantics}

An interpretation is a pair $(\universes, \interpfun)$, where $\universes$ is a set of typed universes and $\interpfun$ is 
an interpretation function, such that for each function symbol $f : \tau_1 \times \cdots \times \tau_n \rightarrow \tau$ in the signature, $\interpfun(f)$ is a concrete function of type $\universes_{\tau_1} \times \cdots \times \universes_{\tau_n} \rightarrow \universes_{\tau}$.
A valuation $\xi$ is a function that maps each variable $x : \tau$ to a member of $\universes_\tau$. 
For a given interpretation $\model$ and valuation $\xi$, we uses $\denotes{t}{\xi}{\model}$ to represent the denotation of $t$ in $\model$ given $\xi$. 
A positive literal $s \eq t$ is true in an interpretation $\model$ for valuation $\xi$ if $\denotes{s}{\xi}{\model} = \denotes{t}{\xi}{\model}$ and false otherwise.
A negative literal $s \noteq t$ is true in an interpretation $\model$ for valuation $\xi$ if $s \eq t$ is false.
A clause $C$ holds in an interpretation $\model$ for valuation $\xi$ if one of its literals is true in $\model$ for $\xi$.
An interpretation $\model$ \emph{models} a clause $C$ if $C$ holds in $\model$ for every valuation.
An interpretation models a clause set, if it models every clause in the set.
A set of clauses $M$ entails a set of clauses $N$, denoted $\holdsin{M}{N}$, if every model of $M$ is also a model of $N$. 
 
\needspace{3\baselineskip}
\section{Calculus}
\label{sec:calculus}

Intuitively, what we are aiming for with our calculus, is that whenever standard superposition applies a substitution $\sigma$ to a conclusion with the side condition 
``$\sigma$ is a unifier of terms $t_1$ and $t_2$'', our calculus adds a constraint $t_1 \noteq t_2$ to the conclusion.
The calculus then has further inference rules that mimic the steps of a first-order unification algorithm and work on negative literals.
Our presentation below does not quite follow this intuition. 
Instead, if the unification problem is trivial we solve it immediately.
If it is non-trivial, we carry out a single step of unification and add the resulting sub-problems as constraints.
Our reasons for doing this are two-fold.

\begin{enumerate}
\item Adding the entire unification problem $t_1 \noteq t_2$ as a constraint can lead to a constraint literal that is larger, with respect to $\succ$, than any literal occurring in the premises. This causes difficulties in the completeness proof.
\item More pertinently, keeping in mind our planned applications to more complex logics, we wish to show that delayed unification remains complete even when only selected sub-problems of the original unification problem are added as constraints. In the context of higher-order logic, for example, this could allow for the eager solving of simple unification sub-problems whilst only the most difficult are added as constraints. See Section \ref{sec:extension} for further details.
\end{enumerate}

Wherever we present a clause as a subclause $C'$ and a literal $l$ (e.g. $C' \vee l$), we denote the entire clause by the same name as the subclause without the dash (e.g. we refer to the clause $C' \vee l$ by $C$).
As in the classical superposition calculus, our calculus is parameterised by a \emph{selection function} that is used to restrict the number of applicable inferences in order to avoid the search space growing unnecessarily.
A selection function $sel$ is a function that maps a clause to a subset of its negative literals. 
We say that literal $l$ is $\sigma$-\emph{eligible} in a clause $C' \lor l$ if it is selected in $C$ ($l \in sel(C)$), or there are no selected literals and $l\sigma$ is maximal in $C\sigma$. 
Strict $\sigma$-eligibility is defined in a like fashion, with maximality replaced by strict maximality. Where $\sigma$ is empty, we sometimes speak of eligibility instead of $\sigma$-eligibility. In what follows, $\constraints$ is a multiset of literals that we refer to as \emph{constraints}.

\begin{center} 
\begin{minipage}[b]{.48\textwidth}

    \begin{scprooftree}{1}
        \def\defaultHypSeparation{\hskip .1in}    
        \AxiomC{$D' \vee f(\tuple{t}_n) \eq t'$}
        \AxiomC{$C' \vee \fosubterm{s}{f(\tuple{s}_n)} \eqOrNorEq s'$}
        \RightLabel{\inf{Sup}}
        \BinaryInfC{$C' \vee D' \vee \fosubterm{s}{t'} \eqOrNorEq s' \vee \constraints$}  
    \end{scprooftree}

  \end{minipage}
  \quad
  \begin{minipage}[b]{.48\textwidth}

    \begin{scprooftree}{1}
        \def\defaultHypSeparation{\hskip .1in}     
        \AxiomC{$D' \vee x \eq t'$}
        \AxiomC{$C' \vee \fosubterm{s}{f(\tuple{s}_n)} \eqOrNorEq s'$}
        \RightLabel{\inf{VSup}}
        \BinaryInfC{$(C' \vee D' \vee \fosubterm{s}{t'} \eqOrNorEq s') \sigma $}  
    \end{scprooftree}

  \end{minipage}
\end{center}

\noindent
where $\sigma = \{ x \rightarrow f(\tuple{s}_n) \}$, and $\constraints = t_1 \noteq s_1 \vee \ldots \vee t_n \noteq s_n$. 
Both rules share the following side conditions. 
Let $t$ stand for either $f(\tuple{t}_n)$ or $x$.
For \inf{Sup}, the substitution $\sigma$ mentioned in the side conditions is of course empty.

\begin{itemize}
\item $t \eq t'$ is strictly $\sigma$-eligible.
\item $\fosubterm{s}{f(\tuple{s}_n)} \eqOrNorEq s'$ is strictly $\sigma$-eligible if positive and $\sigma$-eligible if negative.
\item $t\sigma \not\preceq t'\sigma$ and $\fosubterm{s}{f(\tuple{s}_n)}\sigma \not\preceq s'\sigma$.
\item $C\sigma  \not\preceq D\sigma$
\end{itemize}

\begin{center}
  \begin{minipage}[b]{.48\textwidth}

    \begin{scprooftree}{1}
        \AxiomC{$C' \vee f(\tuple{t}_n) \eq v' \vee f(\tuple{s}_n) \eq v$}
        \RightLabel{\inf{EqFact}}
        \UnaryInfC{$C' \vee v \noteq v' \vee f(\tuple{s}_n) \eq v \vee \constraints $}  
        \singleLine
    \end{scprooftree}

  \end{minipage}
  \quad
  \begin{minipage}[b]{.48\textwidth}

    \begin{scprooftree}{1}
        \AxiomC{$C' \vee u' \eq v' \vee u \eq v$}
        \RightLabel{\inf{VEqFact}}
        \UnaryInfC{$(C' \vee v \noteq v' \vee u \eq v)\sigma $}  
        \singleLine
    \end{scprooftree}

  \end{minipage}
\end{center}

\noindent
for \inf{EqFact}, $\constraints = t_1 \noteq s_1 \vee \ldots \vee t_n \noteq s_n$. For \inf{VEqFact}, either $u$ or $u'$ must be a variable and $\sigma$ is the most general unifier of $u$ and $u'$. 
The side conditions for \inf{EqFact} are:
 
 \begin{itemize}
 \item $f(\tuple{s}_n) \eq v$ be eligible in $C$.
 \item $f(\tuple{s}_n) \not\preceq v$ and $f(\tuple{t}_n) \not\preceq v'$.
 \end{itemize}

\noindent
The side conditions for \inf{VEqFact} are:
 
 \begin{itemize}
 \item $u \eq v$ be $\sigma$-eligible in $C$.
 \item $u\sigma \not\preceq v\sigma$ and $u'\sigma \not\preceq v'\sigma$.
 \end{itemize} 
 
The calculus also contains the following resolution / unification inferences. We refer to these as unification inferences, because each inference represents carrying out a single step of the well-known Robinson unification algorithm \cite{hoder-2009-comparing}.

     \begin{scprooftree}{1}
        \AxiomC{$C' \vee f(\tuple{s}_n) \noteq f(\tuple{t}_n)$}
        \RightLabel{\inf{Decompose}}
        \UnaryInfC{$C' \vee \constraints$}  
        \singleLine
    \end{scprooftree}

\begin{multicols}{2}
    \begin{scprooftree}{1}
        \AxiomC{$C' \vee x \noteq t$}
        \RightLabel{\inf{Bind}}
        \UnaryInfC{$C'\sigma $}  
        \singleLine
    \end{scprooftree}

    \begin{scprooftree}{1}
        \AxiomC{$C' \vee s \noteq s$}
        \RightLabel{\inf{ReflDel}}
        \UnaryInfC{$C'$}  
        \singleLine
    \end{scprooftree}
\end{multicols}

\noindent
where for \inf{Bind}, $\sigma = \{x \rightarrow t\}$ and $x$ does not occur in $t$. For \inf{Decompose}, $f(\tuple{s}_n) \not = f(\tuple{t}_n)$ and $\constraints = t_1 \noteq s_1 \vee \ldots \vee t_n \noteq s_n$. All three inferences require that the final literal be $\sigma$-eligible in $C\sigma$ (for \inf{Decompose} and \inf{ReflDel}, $\sigma$ is empty). We provide some examples to show how the calculus works.

\begin{example}
Consider the unsatisfiable clause set:

\[ C_1 =  f(x, g(x)) \noteq t \qquad C_2 = f(g(b),y) \eq t \]

\noindent

A \inf{Sup} inference between $C_1$ and $C_2$ results in clause $C_3 = t \noteq t \vee x \noteq g(b) \vee g(x) \noteq y$. 
A  \inf{ReflDel} inference on $C_3$ results in the clause $C_4 = x \noteq g(b) \vee g(x) \noteq y$.
An application of \inf{Bind} on $C_4$ with $\sigma = \{ x \rightarrow g(b) \}$ results in $C_5 = g(g(b)) \noteq y$.
Another application of \inf{Bind}, then leads to $\bot$.
\end{example}

\begin{example}
Consider the unsatisfiable clause set:

\[ C_1 = x \eq c \qquad C_2 = f(a,c) \noteq t \qquad C_3 = f(c,c) \eq t \]

\noindent

A \inf{VSup} inference between $C_1$ and $C_2$ results in clause $C_4 = f(c,c) \noteq t$. 
A \inf{Sup} inference between $C_3$ and $C_4$ results in the clause $C_5 = t \noteq t \vee c \noteq c \vee c \noteq c$.
A triple application of \inf{ReflDel} starting from $C_5$ derives $\bot$.

\end{example}

\begin{note}
We abuse terminology and use \emph{inference} and \emph{inference rule} to refer both to schemas such as shown above, 
as well as concrete instances of such schemas. Given an inference $\iota$, we refer to the tuple of its premises by $\prems(\iota)$,
  to its maximal premise by $\mprem(\iota)$, and to its conclusion by $\concl(\iota)$.
\end{note}

\section{Redundancy Criterion}

We utilise Waldmann et al.'s framework \cite{waldmann-2020-comprehensive} for proving the completeness of our calculus. Hence, our redundancy criterion is based on their intersected lifted criterion. 
In instantiating the framework, we roughly follow Bentkamp et al. \cite{bentkamp-2021-lfhol-journal}.
Let the calculus defined above be referred to as $\infsys$. 
We introduce a ground inference system $\ginfsys$ that coincides with standard superposition \cite{bachmair-1994-rewrite}. That is, it contains the well known three inferences, \inf{Sup}, \inf{EqFact} and \inf{EqRes}. 
We refer to these inferences by \inf{GSup}, \inf{GEqFact} and \inf{GEqRes} to indicate that they are only applied to ground clauses.
Following the notation of the framework, we write $\infsys(N)$ ($\ginfsys(N)$) to denote the set of all $\infsys$ ($\ginfsys$) inferences with premises in a clause set $N$.
We introduce a grounding function $\G$ that maps terms, literals and clauses to the sets of their ground instances. For example, given a clause $C$, $\GSigma{C}$ is the set $\{ C\theta \; \vert \;  \theta \text{ is a grounding substitution} \}$. We extend the function $\G$ to clause sets by letting $\GSigmaN = \bigcup_{C \in N} \GSigma{C}$ where $N$ is a set of clauses. 

A ground clause $C$ is redundant with respect to a set of ground clauses $N$ if there are clauses $C_1,\ldots, C_n \in N$ such that for $1 \leq i \leq n$, $C_i \prec C$ and $\holdsin{C_1,\ldots, C_n}{C}$. 
The set of all ground clauses redundant with respect to a set of ground clauses $N$ is denoted $\gredclauses$.

A clause $C$ is redundant with respect to a set of clauses $N$, if for every $D \in \GSigma{C}$, $D$ is redundant with respect to $\GSigmaN$ or there is a clause $C' \in N$ such that $D \in \GSigma{C'}$ and $C \sqsupset C'$ where $\sqsupset$ is the strict subsumption relation. That is $C \sqsupset C'$ if $C$ is subsumed by $C'$, but $C'$ is not subsumed by $C$.
The set of all clauses redundant with respect a set of clauses $N$ is denoted $\redclauses$.

In order to define redundant inferences, we have to pay careful attention to selection functions. For non-ground clauses, we fix a selection function $sel$. We then let $\G(\sel)$ be a set of selection functions on ground clauses with the following property. For each $\gsel \in \G(\sel)$, for every ground clause $C$, there exists a clause $D$ such that $C \in \GSigma{D}$ and the literals selected in $C$ by $\gsel$ correspond to those selected in $D$ by $\sel$. We write $\ginfgsel$ to show that the ground inference system $\ginfsys$ is parameterised by the selection function $\gsel$. Let $\iota$ be an inference in $\infsys$. 
We extend the grounding function $\G$ to a family of grounding functions $\ggsel$ for each $\gsel \in \G(\sel)$. Each function $\ggsel$ maps terms, literals and clauses as above, and maps members of $\infsys$ to subsets of $\ginfgsel$ as follows.\footnote{When a grounding function $\ggsel$ acts on a clause, literal or term, we commonly drop the $\gsel$ superscript as the selection function plays no role in the grounding of these.}

\begin{definition}[Ground Instance of an Inference]
\label{def:inf_redundancy}
Let $\iota$ be of the form $C_1,\ldots, C_n \vdash E \orr \constraints$. An inference $\iota_g \in \ginfgsel$ is in $\ggsel(\iota)$ if it is of the form $C_1\theta,\ldots, C_n\theta \vdash E\theta$ for some grounding substitution $\theta$. In this case, we say that $\iota_g$ is the $\theta$-ground instance of $\iota$. Note that we ignore the constraints in the definition of ground instances. 
\end{definition}

A ground inference $C_1,\ldots,C_n,C \vdash E$ with maximal premise $C$ is redundant with respect to a clause set $N$ if for $1 \leq i \leq n$, $C_i \in \gredclauses$ or $C \in \gredclauses$ or there exist clauses $D_1, \ldots D_m \in N$ such that for $1 \leq i \leq m$, $D_i \prec C$ and $\holdsin{D_1,\ldots, D_m}{E}$. The set of all ground inferences redundant with respect to a set $N$ is denoted $\gredinfs$.

An inference $\iota$ is redundant with respect to a clause set $N$ if for every $\gsel \in \G(\sel)$ and for every $\iota' \in \ggsel(\iota)$, $\iota' \in \mathit{GRed^{\gsel}_I}(\GSigmaN)$. In words, every ground instance of the inference is redundant with respect to $\GSigmaN$. We denote the set of all redundant inferences with respect to a set $N$ as $\redinfs$. 

A clause set $N$ is saturated up to redundancy by an inference system $\infsys$ if every member of $\infsys(N)$ is redundant with respect to $N$.

\begin{note}
Given the definition of clause redundancy above, the \inf{ReflDel} inference can be utilised as a \emph{simplification} inference.
That is, the conclusion of the inference renders the premise redundant.
\end{note}

\section{Refutational Completeness}
\label{sec:completeness_proof}

To prove refutational completeness we utilise the above mentioned framework of Waldmann et al. \cite{waldmann-2020-comprehensive}. 
In particular, we use Theorem 14 from the paper to lift completeness from the ground level to the non-ground level. We bring Theorem 14 here
for clarity and to keep the paper self contained. We then present it in our notation. 
Let $\mathit{GRed} = (\gredi,\gredc)$ and $\mathit{Red} =  (\redi,\redc)$

\setcounter{theorem}{13}
\begin{theorem}[from Waldmann et al. \cite{waldmann-2020-comprehensive}] 
\label{thm:waldmann_lifting}
If $(\mathit{GInf^q},\mathit{Red^q})$ is statically refutationally complete w.r.t. $\models^q$ for every $q \in Q$
and if for every $N \subseteq \mathbf{F}$ that is saturated w.r.t. $\mathit{FInf}$ and $\mathit{Red}^{\cap \G}$ there exists a $q$ such that
$\mathit{GInf^q}(\G^q(N)) \subseteq \G^q(\mathit{FInf}(N)) \cup \mathit{Red^{q}_I}(\G^q(N))$, then 
$(\mathit{FInf}, \mathit{Red}^{\cap \G})$ is statically refutationally complete w.r.t. $\models^{\cap}_\G$.
\end{theorem}

\setcounter{theorem}{13}
\begin{theorem}[from Waldmann et al. in our Notation] 
\label{thm:waldmann_lifting}
If $(\ginfgsel,\mathit{GRed})$ is statically refutationally complete w.r.t. $\models$ for every $\gsel \in \G(\sel)$
and if for every clause set $N$ that is saturated w.r.t. $\infsys$ and $\mathit{Red}$ there exists a $\gsel$ such that
$\ginfgsel(\ggsel(N)) \subseteq \ggsel(\infsys(N)) \cup \redi(\ggsel(N))$, then 
$(\infsys, \mathit{Red})$ is statically refutationally complete w.r.t. $\models_\G$.
\end{theorem}
\setcounter{theorem}{0}

Thus, in our context, the set $Q$ is $\G(\sel)$, the ground inference system $\mathit{GInf^q}$ maps to $\ginfgsel$, the ground redundancy criterion
$\mathit{Red^q}$ maps to $( \gredi,\gredc)$ and the ground entailment relation $\models^q$ maps to standard entailment on first-order
clauses. 
Moreover, the non-ground inference system $\mathit{FInf}$ maps to $\infsys$ and the redundancy criterion $\mathit{Red}^{\cap \G}$
maps to $(\redi, \redc)$. Note, that this final mapping is not exact, as the criterion $\mathit{Red}^{\cap \G}$ does not allow for a tiebreaker ordering,
such as the strict subsumption relation, to be utilised in the definition of non-ground redundancy. However, this mismatch can easily be repaired
since Theorem 16 of the framework paper extends the result of Theorem 14 to the case where tiebreaker orderings are used. 

As our ground inference systems $\ginfgsel$ are ground superposition systems, static refutational completeness with respect to standard entailment
and standard redundancy is a famous result. See for example \cite{bachmair-1990-restrictions}. What remains for us to prove in order to apply Theorem 14 and show the static refutational completeness of $\infsys$, is:

\begin{enumerate}
\item For every $\gsel \in \G(\sel)$, the grounding function $\ggsel$ is a grounding function in the sense of the framework.
\item For every clause set $N$ saturated up to redundancy by $\infsys$, there exists a $\gsel \in \G(\sel)$ such that $\ginfgsel$ $(\G(N)) \subseteq \ggsel(\infsys(N)) \cup \gredi(\G(N))$.
In words, there exists a ground selection function such that every ground inference with that selection function and premises in $\G(N)$ 
is either the instance of a non-ground inferences with premises in $N$ or is redundant with respect to $\G(N)$.
\end{enumerate}

\begin{lemma}
\label{lem:grounding}
For every $\gsel \in \G(\sel)$, the grounding function $\ggsel$ is a grounding function in the sense of the framework.
\end{lemma}

\begin{proof}
We need show that properties (G1) -- (G3) defined by Waldmann et al. hold for grounding functions. These properties are:

\begin{itemize}[leftmargin=35pt]
\item[(G1)] for every $\bot \in \mathbf{F}_{\bot}$, $\emptyset \not= \G(\bot) \subseteq \mathbf{G}_{\bot}$; 
\item[(G2)] for every $C \in \mathbf{F}$, if $\bot \in \G(C)$ and $\bot \in \mathbf(G)_{\bot}$ then $C \in \mathbf{F}_{\bot}$;
\item[(G3)] for every $\iota \in \mathit{FInf}$, if $\G(\iota) \not= \mathit{undef}$, then $\G(\iota) \subseteq \mathit{Red}_I(\G(\concl(\iota)))$.
\end{itemize}

As properties (G1) and (G2) relate to
the grounding of terms and clauses, and our grounding of these is fully standard we skip these. We prove (G3), which in our terminology is: for every $\iota \in \infsys$,
$\ggsel(\iota) \subseteq \gredi(\G(\mathit{concl}(\iota)))$. 
This can be achieved by showing that for every $\iota' \in \ggsel(\iota)$, there exist clauses $\tuple{C} \in \G(\concl(\iota))$ such that $\holdsin{\tuple{C}}{\concl(\iota')}$ and for each $C_i \in \tuple{C}$, $C_i \prec \mprem(\iota')$. In what follows, let $\theta$ be the substitution by which $\iota'$ is a grounding of $\iota$.

If $\constraints$ is the empty set in $\concl(\iota)$, then $\concl(\iota)\theta = \concl(\iota')$ and hence $\holdsin{\concl(\iota)\theta}{\concl(\iota')}$. Moreover, $\concl(\iota)\theta \in \G(\concl(\iota))$ and  $\concl(\iota)\theta \prec \mprem(\iota')$. 

On the other hand, if $\constraints$ is not empty, let $u = f(\tuple{t}_n)$ and $u' = f(\tuple{s}_n)$ be the two terms within $\prems(\iota)$ from which the constraints are created. 
By the existence of $\iota'$, we have that $u\theta = u'\theta$, and hence that $t_i\theta = s_i\theta$ for $1 \leq i \leq n$. 
Hence, every literal in $\constraints\theta$ has the form $t \noteq t$ and is trivially false in every interpretation. 
Thus, we still have $\holdsin{\concl(\iota)\theta}{\concl(\iota')}$.
Moreover, by the subterm property of the ordering $\succ$ we have that $t_i\theta \noteq s_i\theta$ is smaller than the maximal / selected literal of 
$\mprem(\iota')$ for $1 \leq i \leq n$ and hence that $\concl(\iota)\theta \prec \mprem(\iota')$. \qed
\end{proof}

\begin{lemma}
\label{lem:useful}
let $\sigma$ be the most general unifier of terms $s$ and $s'$, and $\theta$ be any unifier of the same terms. Then for any term $t$, $(t\sigma)\theta = t\theta$.
\end{lemma}
\begin{proof}
Since $\sigma$ is the most general unifier, there must be a substitution $\rho$ such that $\sigma\rho = \theta$. Hence $(t\sigma)\theta = (t\sigma)\sigma\rho = t\sigma\rho = t\theta$ where the second to last step follows from the fact that $\sigma$ is idempotent. \qed
\end{proof}

\begin{lemma}
\label{lem:liftable}
For every clause set $N$ saturated by $\infsys$, there exists a $\gsel \in \G(\sel)$ such that $\ginfgsel$ $(\G(N)) \subseteq \ggsel(\infsys(N)) \cup \gredi(\G(N))$.
\end{lemma}
\begin{proof}
For every $D \in \G(N)$ there must exist a clause $C \in N$ such that $D \in \G(C)$. Let $\gg$ be an arbitrary well-founded ordering on clauses.
We let $C = \Ginv(D)$ denote the $\gg$-smallest clause such that $D \in \G(C)$. We then choose the $\gsel \in \G(\sel)$ that for a clause $D \in \G(N)$ selects the corresponding literals to those selected by $\sel$ in $\Ginv(D)$. 
Given this $\gsel$, we need to show that every inference with premises in $\G(N)$ is either the ground instance of an inference with premises in $N$,
or is redundant with respect to $\G(N)$.

A \inf{Sup} inference is redundant if the term $t$ replaced in the second premise occurs at or below a variable.
The proof is exactly the same as in the standard proof of the completeness of superposition \cite{bachmair-1994-rewrite}, so we don't repeat it.
All other inferences can be shown to be the ground instance of inferences from clauses in $N$.

Let $\iota \in \ginfgsel$ be the following \inf{GSup} inference with premises in $\G(N)$.

    \begin{scprooftree}{1}
        \AxiomC{$D'\theta \vee t\theta \eq t'\theta$}
        \AxiomC{$C'\theta \vee \fosubterm{s\theta}{t\theta} \eqOrNorEq s'\theta$}
        \RightLabel{}
        \BinaryInfC{$C'\theta \vee D'\theta \vee \fosubterm{s\theta}{t'\theta} \eqOrNorEq s'\theta$}  
    \end{scprooftree}

\noindent
where $\Ginv(D\theta) = D = D' \orr t \eq t'$, $\Ginv(C\theta) = C = C' \orr s \eqOrNorEq s'$ and $\iota$ fulfils all the side conditions of \inf{GSup}. 
Let $\sigma$ be any substitution. 
The literal $t\theta \eq t'\theta$ being strictly maximal in $D\theta$ implies that $t\sigma \eq t'\sigma$ is strictly maximal in $D\sigma$ due to the stability under substitution of $\succ$.
The literal $\fosubterm{s\theta}{t\theta} \eqOrNorEq s'\theta$ being (strictly) eligible in $C\theta$ with respect to $\gsel$ implies that $s\sigma \eq s'\sigma$ is strictly eligible in $C\sigma$ with respect to $\sel$. 
Let $p$ be the position of $t\theta$ within $s\theta$ and let $u$ be the subterm of $s$ at $p$.
Since the term $t\theta$ does not occur below a variable of $C$, such a position must exist. 
Moreover, $u$ cannot be a variable since if it was $t\theta$ would occur at a variable of $C$.
As $\theta$ is a unifier of $u$ and $t$, it must be the case that either $t$ is a variable, or $u$ and $t$ have the same top symbol.
Further, $D\theta \prec C\theta$ implies that $C\sigma  \not\preceq D\sigma$, 
$t\theta \succ t'\theta$ implies that $t\sigma  \not\preceq t'\sigma$, and
$\fosubterm{s\theta}{t'\theta} \succ s'\theta$ implies $s\sigma  \not\preceq s'\sigma$.
Thus, if $t$ is not a variable, there exists the following \inf{Sup} inference $\iota'$ from clauses $D$ and $C$.

    \begin{scprooftree}{1}
        \AxiomC{$D' \orr t \eq t'$}
        \AxiomC{$C' \orr \fosubterm{s}{u} \eqOrNorEq s'$}
        \RightLabel{}
        \BinaryInfC{$C' \orr D' \orr \fosubterm{s}{t'} \eqOrNorEq s' \orr \constraints$}  
    \end{scprooftree}

We have that $(C' \orr D' \orr \fosubterm{s}{t'} \eqOrNorEq s')\theta = \concl(\iota)$.
That is, the grounding of the conclusion of $\iota'$ less the constraint literals is equal to the conclusion of $\iota$.
Thus, $\iota$ is the $\theta$-ground instance of $\iota'$ as per Definition \ref{def:inf_redundancy} .
If $t$ is a variable $x$, then there exists the following \inf{VSup} inference $\iota'$ from clauses $D$ and $C$.

    \begin{scprooftree}{1}
        \AxiomC{$D' \orr x \eq t'$}
        \AxiomC{$C' \orr \fosubterm{s}{u} \eqOrNorEq s'$}
        \RightLabel{}
        \BinaryInfC{$(C' \orr D' \orr \fosubterm{s}{t'} \eqOrNorEq s')\sigma$}  
    \end{scprooftree}

Where $\sigma = \{ x \rightarrow u \}$ is the most general unifier of $t$ and $u$. Thus, we can use Lemma \ref{lem:useful} to show that $\concl(\iota')\theta = \concl(\iota)$ and again $\iota$ is the $\theta$-ground instance of $\iota'$. 

Let $\iota \in \ginfgsel$ be the following \inf{GEqFact} inference with premise in $\G(N)$.

    \begin{scprooftree}{1}
        \AxiomC{$C'\theta \orr u'\theta \eq v'\theta  \vee u\theta  \eq v\theta $}
        \RightLabel{}
        \UnaryInfC{$C'\theta \orr v\theta \noteq v'\theta \vee u\theta \eq v\theta $}  
        \singleLine
    \end{scprooftree}

\noindent
where $u'\theta = u\theta$, $\Ginv(C\theta) = C = C' \orr u' \eq v' \orr u \eq v$ and $\iota$ fulfils all the side conditions of \inf{GEqFact}.
Let $\sigma$ be any substitution. 
The literal $u\theta  \eq v\theta$ being maximal in $D\theta$ implies that $u\sigma  \eq v\sigma$ is maximal in $D\sigma$.
Since $\theta$ is a unifier of $u'$ and $u$, at least one of them must be a variable, or they must share a top symbol.
Moreover, $u\theta  \succ v\theta$ implies that $u\sigma  \not\preceq v\sigma$ and $u'\theta  \succ v'\theta$ implies that $u'\sigma  \not\preceq v'\sigma$.
If neither $u$ nor $u'$ is a variable, there exists the following \inf{EqFact} inference $\iota'$ from $C$.

    \begin{scprooftree}{1}
        \AxiomC{$C' \orr u' \eq v'  \orr u  \eq v$}
        \RightLabel{}
        \UnaryInfC{$C' \vee v \noteq v' \vee u \eq v \vee \constraints $}       
        \singleLine
    \end{scprooftree}
    
We have $(C' \vee v \noteq v' \orr u \eq v)\theta = \concl(\iota)$, making $\iota$ the $\theta$-ground instance of $\iota'$ as per Definition \ref{def:inf_redundancy}. If either $u$ of $'u$ is a variable there exists the following \inf{VEqFact} inference $\iota'$ from $C$.

    \begin{scprooftree}{1}
        \AxiomC{$C' \orr u' \eq v'  \orr u  \eq v$}
        \RightLabel{}
        \UnaryInfC{$(C' \vee v \noteq v' \vee u \eq v)\sigma $}       
        \singleLine
    \end{scprooftree}

Where $\sigma$ is the most general unifier of $u$ and $u'$. Thus, we can use Lemma \ref{lem:useful} to show that $\concl(\iota')\theta = \concl(\iota)$.
Finally, let $\iota \in \ginfgsel$ be the following \inf{GEqRes} inference with premise in $\G(N)$.

    \begin{scprooftree}{1}
        \AxiomC{$C'\theta \orr s\theta \noteq s'\theta$}
        \RightLabel{}
        \UnaryInfC{$C'\theta$}  
        \singleLine
    \end{scprooftree}

where $s\theta = s'\theta$, $\Ginv(C\theta) = C = C' \orr s \noteq s'$ and $\iota$ fulfils all the side conditions of \inf{GEqRes}.
Let $\sigma$ be any substitution. 
The literal $s\theta \noteq s'\theta$ being eligible with respect to $\gsel$ in $C\theta$ implies that $s \noteq s'$ is eligible in $C$ with respect to $\sel$.
Since $\theta$ is a unifier of $s$ and $s'$, at least one of them must be a variable, or they must share a top symbol.
If $s = s'$, then there exists the following \inf{ReflDel} inference  $\iota'$ from $C$.

    \begin{scprooftree}{1}
        \AxiomC{$C' \vee s \noteq s$}
        \RightLabel{}
        \UnaryInfC{$C'$}  
        \singleLine
    \end{scprooftree}

Otherwise we have two options. If either $s$ (or analogously $s'$) is a variable, then there is the following \inf{Bind} inference  $\iota'$ from $C$.

    \begin{scprooftree}{1}
        \AxiomC{$C' \vee x \noteq s'$}
        \RightLabel{}
        \UnaryInfC{$C'\sigma$}  
        \singleLine
    \end{scprooftree}

Otherwise $s$ and $s'$ must share a top symbol and there is the following \inf{Decompose} inference $\iota'$ from $C$.

    \begin{scprooftree}{1}
        \AxiomC{$C' \vee f(\tuple{s}_n) \noteq f(\tuple{t}_n)$}
        \RightLabel{}
        \UnaryInfC{$C' \orr \constraints$}  
        \singleLine
    \end{scprooftree}

In the first case, we have $\concl(\iota')\theta = \concl(\iota)$. 
In the second case, $\sigma$ is the most general unifier of $s$ and $s'$, so we can use Lemma \ref{lem:useful} to show that $\concl(\iota')\theta = \concl(\iota)$.
In the last case, we have that $C'\theta = \concl(\iota)$.
Thus in all cases, $\iota$ is the $\theta$-ground instance of $\iota'$. \qed
\end{proof}

Using Lemmas \ref{lem:grounding} and \ref{lem:liftable} we can instantiate Theorem \ref{thm:waldmann_lifting} to prove the static refutational completeness of $\infsys$. 
There is a slight issue here, as Theorem \ref{thm:waldmann_lifting} gives us refutational completeness with respect to Herbrand entailment. 
That is $N \models M$ if $\G(N) \models \G(M)$. 
We would like to prove completeness with respect to entailment as defined in Section \ref{sec:tech_background} (known as Tarski entailment).
This issue can easily be resolved by showing that the two concepts are equivalent with regards to refutations which can be achieved in a manner similar to Bentkamp et al. (Lemma 4.19 of \cite{bentkamp-2021-lfhol-journal}). 

\begin{theorem}[Static refutational completeness]
For a set of clauses $N$ saturated up to redundancy by $\infsys$, $\holdsin{N}{\bot}$ if and only if $\bot \in N$.
\end{theorem}

Theorem 17 of Waldmann et al.'s framework can be used to derive dynamic refutational completeness from static refutational completeness.
We refer readers to the framework for the formal definition of dynamic refutational completeness.

\begin{theorem}[Dynamic refutational completeness]
The inference system $\infsys$ is dynamically refutationally complete with respect to the redundancy criterion $(\redi, \redc)$.
\end{theorem}

\section{Extending to Higher-Order Logic}
\label{sec:extension}

We sketch how the ideas above can be extended to higher-order logic. 
This is ongoing research, and many of the technical details have yet to be fully worked out. 
Here, we provide a (very) informal description and then provide examples. 
The higher-order unification problem is undecidable and there can exist a potentially infinite number of incomparable most general unifiers for a pair of terms \cite{huet-1975-unif-algo}. 
Existing higher-order paramodulation style calculi deal with this issue in two main ways. One method is to abandon completeness and only unify to some predefined depth \cite{steen-2018-higher}. 
Another approach is to produce potentially infinite streams of unifiers and interleave the fetching of items from such streams with the standard saturation procedure\cite{bentkamp-et-al-2019-lhol}.  
Our idea is to solve easy sub-problems eagerly, such as when terms are first-order or in the pattern fragment \cite{nipkow-1993-functional}, and add harder sub-problems as constraints. 
We then utilise dedicated inferences on negative literals to mimic the rules of Huet's well known (pre-)unification procedure \cite{huet-1975-unif-algo}. 
We think that inferences similar to the following two, could be sufficient to achieve refutational completeness.

     \begin{scprooftree}{1}
        \AxiomC{$C' \vee x \, \tuple{s}_n \noteq f \, \tuple{t}_m $}
        \RightLabel{\inf{Imitate}}
        \UnaryInfC{$(C' \vee x \, \tuple{s}_n \noteq f \, \tuple{t}_m) \{ x \rightarrow \lambda \tuple{y}_n. \, f \, \tuple{(z_1 \,  \tuple{y}_n)}_m   \} $}  
        \singleLine
    \end{scprooftree}

     \begin{scprooftree}{1}
        \AxiomC{$C' \vee x \, \tuple{s}_n \noteq f \, \tuple{t}_m $}
        \RightLabel{\inf{Project}}
        \UnaryInfC{$(C' \vee x \, \tuple{s}_n \noteq f \, \tuple{t}_m) \{ x \rightarrow \lambda \tuple{y}_n. \, y_i \, \tuple{(z_1 \,  \tuple{y}_n)}_p   \} $}  
        \singleLine
    \end{scprooftree}

\noindent
In both rules, each $z_i$ is a fresh variable of the relevant type, and  $x \, \tuple{s}_n \noteq f \, \tuple{t}_m$ is selected in $C$.
\inf{Project} has $k \leq n$ conclusions, one for each $y_i$ of suitable type.
We hope that through a careful definition of the selection function, along with the use of purification, we can avoid the need to apply unification inferences to flex-flex literals (negative literals where both sides of the equality have variable heads). 
Moreover, we are hopeful that the calculus we propose can remain complete without the need for inferences that carry out superposition beneath variables such as the \inf{FluidSup} rule of $\lambda$-superpostion \cite{bentkamp-et-al-2019-lhol} and the \inf{SubVarSup} rule of combinatory-superposition \cite{bhayat-reger-2020-comb-sup}.

\begin{example}
Consider the unsatisfiable clause set:

\[ C_1 =  f \, y \, (x \, a) \, (x \, b) \noteq t \qquad C_2 = f \, c \, a \, b \eq t \]

\noindent
A \inf{Sup} inference between $C_1$ and $C_2$ results in clause $C_3 = t\sigma \noteq t\sigma \vee x \, a \noteq a \vee x \, b \noteq b$ where $\sigma = \{ y \rightarrow c \}$. Assume that the literal $x \, a$ is selected in $C_3$. We can carry out either a \inf{Project} step on this literal or an \inf{Imitate} step. The result of a project step is $C_4 = (t\sigma \noteq t\sigma \vee (\lambda z.\,z) \, a \noteq a \vee x \, b \noteq b)\{x \rightarrow \lambda z.\,z \}$. Applying the substitution and $\beta$-reducing results in $C_5 = t\sigma \noteq t\sigma \vee  a \noteq a \vee b \noteq b$ from which it is easy to reach a contradiction.

\end{example}

\begin{example}[Example 1 of Bentkamp et al. \cite{bentkamp-et-al-2019-lhol}]
Consider the unsatisfiable clause set:

\[ C_1 =  f \, a \eq c \qquad C _2  = h \, (y \, b) \, (y \, a) \noteq h \, (g \, (f \, b)) \, (g \, c) \]

\noindent
An \inf{EqRes} inference on $C_2$ results in $C_3 = y \, b \noteq g \, (f \, b) \vee y \, a \noteq g \, c$. 
An \inf{Imitate} inference on the first literal of $C_3$ followed by the application of the substitution and some $\beta$-reduction results in $C_4 = g \, (z \, b) \noteq g \, (f \, b) \vee g \, (z \, a) \noteq g \, c$. 
A further double application of \inf{EqRes} gives us $C_5 = z \, b \noteq f \, b \vee z \, a \noteq c$.
We again carry out \inf{Imitate} on the first literal followed by an \inf{EqRes} to leave us with $C_6 = x \, b \noteq b \vee f \, (x \, a) \noteq c$.
We can now carry out a \inf{Sup} inference between $C_1$ and $C_6$ resulting in $C_7 =  x \, b \noteq b \vee c \noteq c \vee x \, a \noteq a$ from which it is simple to derive $\bot$ via an application of \inf{Imitate} on either the first or the third literal. 
Note, that the empty clause was derived without the need for an inference that simulates superposition underneath variables, unlike in \cite{bentkamp-et-al-2019-lhol}.
\end{example}

\begin{example}[Example 2 of Bentkamp et al. \cite{bentkamp-et-al-2019-lhol}]
Consider the unsatisfiable clause set:

\[ C_1 =  f \, a \eq c \qquad C _2  = h \, (y \, (\lambda x. \, g \, (f \, x)) \, a) \, y \noteq h \, (g \, c) \, (\lambda w \, x. \, w \, x)   \]
An \inf{EqRes} inference on $C_2$ results in $C_3 = y \, (\lambda x. \, g \, (f \, x)) \, a \noteq g \, c \vee y \noteq \lambda w \, x. \, w \, x$. 
Assuming that the second literal is selected,\footnote{Most orderings would select the first literal. In this case, we can still derive a contradiction, but the proof is longer.} an \inf{EqRes} inference results in $C_4 = (y \, (\lambda x.$ $ \, g \, (f \, x)) \, a \noteq g \, c)\{y \rightarrow \lambda w \, x. \, w \, x \}$.
Simplifying $C_4$ via applying the substitution and $\beta$-reducing, we achieve $g \, (f \, a) \noteq g \, c$. 
Superposing $C_1$ onto this clause we end up with $C_5 = g \, c \noteq g \, c$ from which the empty clause can easily be derived.
Note again, that the empty clause has been derived without recourse to a \inf{FluidSup}-like inference.
\noindent

\end{example}

\section{Experimental Results}
\label{sec:results}
We implemented the calculus in the Vampire theorem prover \cite{vampire}. 
We also implemented a variant of the calculus, that utilises fingerprint indices \cite{schulz-2012-fingerprint} to act as an imperfect filter.
The completeness proof indicates that a superposition inference only needs to be carried out when the two terms can \emph{possibly} unify.
Therefore, we store terms in fingerprint indices, which act as fast imperfect filters for finding unification partners, and only carry out superposition inferences with terms returned by the index. 
This restricts, somewhat, the number of inferences that take place, at the expense of some loss of speed.
Thus, it represents a midway path between eager unification and delayed unification.
As a final twist, we implemented a version of the calculus that uses fingerprint indices as well as solving constraint literals of the form $x \noteq t$ (where $x$ is not a subterm of $t$) and $t \noteq t$ eagerly.
Thus, in this version of the calculus there is no need for the \inf{Bind} and \inf{ReflDel} rules. 

We compared each of these approaches with the standard superposition calculus implemented in Vampire.
We refer to the standard calculus as \inf{Vampire} and the delayed inference calculus without fingerprint indices by \inf{Vampire*}. \footnote{Our implementation can be found at \url{https://github.com/vprover/vampire/tree/delayed-unification}. To run the new calculus, use option \texttt{-duc on}. To run the standard calculus, the option \texttt{duc} is set to \texttt{off}. }
We refer to the delayed inference calculus with fingerprint indices by \inf{Vampire\textsuperscript{$\dagger$}}. 
Finally, we refer to the calculus that eagerly solves some constraint literals by \inf{Vampire\textsuperscript{$\ddagger$}}.\footnote{The code for both  \inf{Vampire\textsuperscript{$\dagger$}} and \inf{Vampire\textsuperscript{$\ddagger$}} can be found at branch \url{https://github.com/vprover/vampire/tree/delayed-unif-with-fp}. \inf{Vampire\textsuperscript{$\dagger$}}  was built from commit \texttt{c04a08feb5db3e7468a1fa} and \inf{Vampire\textsuperscript{$\ddagger$}} from commit \texttt{fa2f139302b6a7a6487e73}. Again, option \texttt{-duc on} is required for the new calculi to run.}

We tested these approaches against each other on benchmarks coming from CASC 2023 system competition \cite{casc}.
As our new approach is not currently compatible with higher-order or polymorphic input, we restricted the comparison to monomorphic first-order problems.
Namely, we used the 500 benchmarks in the FNE and FEQ categories. 
These are monomorphic, first-order benchmarks that either include equality (FEQ) or do not contain equality (FNE).
All benchmarks in the set are theorems.
The results can be seen in Table \ref{tab:results}. 
All experiments were run on a node cluster located at The University of Manchester. 
Each node in the cluster is equipped with 192 gigabytes of RAM and 32 Intel\textsuperscript{\textregistered} Xeon\ processors with two threads per core.
Each configuration was given 100s of CPU time per problem and run in single core mode.
\inf{Vampire} was run with options \texttt{-\--mode casc} which causes it to use a tuned portfolio of strategies.
All other variants were run with options \texttt{-\--mode casc -\--forced\_options duc=on} which forces the use of the new calculus on top of the aforementioned portfolio.

\setlength{\tabcolsep}{15pt} 
\begin{table}[]
\centering
\begin{tabular}{@{}lcc@{}}
\hline
  \textbf{Approach} & \textbf{Solved} & \textbf{Uniques}    \\ \hline
  \inf{Vampire} & 430   &  110  \\ \hline 
   \inf{Vampire*} & 238 &    0 \\ \hline
 \inf{Vampire\textsuperscript{$\dagger$}} & 255 &   0  \\ \hline
 \inf{Vampire\textsuperscript{$\ddagger$}} & 322 & 2 \\ \hline
\end{tabular}
\caption{Summary of experimental results}
\label{tab:results}
\end{table}

The calculi based on delayed unification perform badly in comparison to standard superposition. 
This is unsurprising, as syntactic first-order unification is already an efficient process. 
By replacing it with delayed unification, we gain little in terms of time, but pay a heavy penalty in terms of the number
of inferences carried out. 
The use of fingerprint indices helps somewhat in mitigating this issue, but not a great deal.
Eagerly solving trivial constraints shows more promise and is actually able to solve two problems that the standard calculus can not (within the time limit). 
These are the benchmarks \texttt{CSR036+3.p} and \texttt{LAT347+3.p}.

\section{Related Work}
The only other proof calculi that we are aware of that explicitly integrate unification rules at the calculus level, are the higher-order paramodulation calculi \cite{benzmuller-2015-higher,steen-2018-higher} and lazy paramodulation~\cite{lazy}.
However, these calculi are paramodulation calculi and do not incorporate certain concepts of redundancy so crucial to the success of superposition provers.
Moreover, the completeness proofs for these calculi are based on very different techniques to the Bachmair \& Ganzinger style model building proofs commonly employed in the completeness proofs of superposition calculi.

There are other calculi that in some form do represent the folding of unification into the calculus, but the link between the unification rules and the calculus is less clear. 
For example, the recent work by one of the authors of this paper \cite{schoisswohl-2023-ALASCA} relating to reasoning about linear arithmetic, moves theory reasoning relating to a number of equations from the unification algorithm to the calculus level. 
A different example, by another of this paper, is the combinatory-superposition calculus \cite{bhayat-reger-2020-comb-sup} which essentially folds higher-order combinatory unification into the calculus. 
In both cases, the relationship between the unification algorithm and the calculus rules is not obvious.

There are other methods of dovetailing unification with inference rules. 
For example, a unification procedure can be modified to return a stream of results.
This stream can be interrupted in order to carry out further inferences and then returned to later.
This is the approach taken by the higher-order Zipperposition prover \cite{bentkamp-et-al-2019-lhol} in order to handle the infinite sets of unifiers returned by higher-order unification.
Conceptually, this is a very different solution to using constraints, since the intermediate terms created during unification are not available to the entire calculus as they are in our approach. 
Furthermore, from an implementation perspective, streams of unifiers are a far greater departure from the standard saturation architecture than the adding of constraints.
Unification can also be partially delayed by preprocessing techniques such as Brand's modification method and its developments~\cite{cee}.

As mentioned in the introduction, abstraction resembles the basic strategy \cite{bachmair-1992-basic,nieuwenhuis-1992-basic-complete}, 
where unification problems are added to the constraint part of a clause.
Periodically, these constraints can be checked for satisfiability and clauses with unsatisfiable constraints removed.
However, in the basic strategy, the constraints do not interact with the rest of the proof calculus.
Moreover, redundancy of clauses can no longer be defined in terms of ground instances, but only in terms of ground instances that satisfy the constraints.
This significantly affects the simplification machinery of superposition / resolution.

Unification with abstraction was first introduced, to the best of our knowledge, by Reger et al. in \cite{reger-unif-with-abstarction-2018} in the context of theory reasoning.
However, the concept was introduced in an ad-hoc fashion with no theoretical analysis of its impact on the completeness of the underlying calculus.
Recently, the relationship between unification modulo an equational theory and unification with abstraction has been analysed \cite{schoisswohl-2023-ALASCA} and a framework developed linking the two. 
It remains to explore whether the current work can fit into that framework.

\section{Conclusion}

We have developed a first-order superposition calculus that delays unification through the use of constraints, and proved its completeness.
Whilst the calculus does not perform well in practice, we feel that the calculus and its completeness proof form a template that can be followed to prove the completeness of calculi that involve unification procedures more complex than syntactic first-order unification.
For example unification modulo a set of equations $E$. 
Some of the crucial features of our approach are: 
(1) the carrying out of partial unification and adding the remaining unification pairs back as constraints, and (2)
the ignoring of constraint literals in the definition of redundant inference.
In particular, feature (1) may well be crucial in taming issues relating to undecidable unification problems.
For example, in higher-order logic where unification is undecidable, it is common to run unification to a particular depth and then give up if termination has not occurred.
Of course, this harms completeness.
With our approach it should be possible to add the remaining unification pairs back as constraints and maintain completeness.
In the future, we would like to generalise our approach into a framework that can be used to prove the completeness of a variety of calculi as long as the unification problem for the underlying terms meets certain conditions.
We would also like to explore instantiating such a framework to prove the completeness of particular calculi of interest to us such as AC-superposition and higher-order superposition.

\smallskip

\noindent{\bf Acknowledgements.} \\
We  acknowledge funding from 
the ERC Consolidator Grant ARTIST 101002685,
the TU Wien Doctoral College SecInt, and 
the FWF SFB project SpyCoDe F8504.

\bibliographystyle{splncs04}
\bibliography{references}

\end{document}